\documentclass[12pt]{article}
\usepackage[margin=2cm]{geometry}
\usepackage{amsmath,amsfonts,amssymb,amstext,amsthm,amssymb,bbm,hyperref}

\theoremstyle{plain}
\newtheorem{thm}{Theorem}
\newtheorem{defin}{Definition}[section]

\newtheorem{rmk}[defin]{Remark}

\newtheorem{cor}[defin]{Corollary}
\newtheorem{lemma}[defin]{Lemma}
\newtheorem*{thm*}{Theorem}

\newtheorem{manualtheoreminner}{}
\newenvironment{manualtheorem}[1]{%
  \renewcommand\themanualtheoreminner{#1}%
  \manualtheoreminner
}{\endmanualtheoreminner}

\def\mes{\operatorname{mes}}
\def\Res{\operatorname{Res}}
\def\dist{\operatorname{dist}}
\def\coupling{{g}}

\begin{document}

\title{Upper bounds on quantum dynamics in arbitrary dimension}
\author{Mira Shamis\textsuperscript{1}  and Sasha Sodin\textsuperscript{1,2} }
\footnotetext[1]{School of Mathematical Sciences,
Queen Mary University of London, 
Mile End Road, London E1 4NS, England. email: [m.shamis, a.sodin]@qmul.ac.uk.}
\footnotetext[2]{Supported in part by the European Research Council starting grant 639305 (SPECTRUM), a Royal Society Wolfson Research Merit Award (WM170012), and a Philip Leverhulme Prize of the Leverhulme Trust (PLP-2020-064).}
\maketitle

\begin{abstract} Motivated by the research on upper bounds on the rate of quantum transport for one-dimensional operators, particularly, the recent works of Jitomirskaya--Liu and Jitomirskaya--Powell and the earlier ones of Damanik--Tcheremchantsev, we  propose a method to prove similar bounds in arbitrary dimension. The method applies  both to Schr\"odinger and to long-range operators.

In the case of ergodic operators, one can use  large deviation estimates for the Green function in finite volumes to verify the assumptions of our general theorem. Such estimates have been proved for numerous classes of quasiperiodic operators in one and higher dimension, starting from the works of Bourgain, Goldstein, and Schlag.

One of the applications is a power-logarithmic bound on  the quantum transport defined by a multidimensional discrete Schr\"odinger (or even long-range) operator associated with an irrational shift, valid for all Diophantine frequencies  and uniformly for all phases. To the best of our knowledge, these are the first  results on the quantum dynamics for quasiperiodic operators in dimension greater than one that do not require exclusion of a positive measure of phases. Moreover, and in contrast to localisation, the estimates are uniform in the phase. 

The arguments are also applicable to  ergodic operators corresponding to other kinds of base dynamics, such as the skew-shift.

\end{abstract}

\section{Introduction} 

\paragraph{Background} Consider a bounded self-adjoint operator $H$ acting on $\ell_2(\mathbb Z^\nu)$ as a sum of a convolution  with  a symmetric kernel $A: \mathbb Z^\nu \to \mathbb R$ and multiplication  by  a potential $V: \mathbb Z^\nu \to \mathbb R$: 
\begin{equation}\label{eq:defop} (H \phi)(w) = (A * \phi + V \phi)(w) = \sum_{v \in \mathbb Z^\nu} A(w-v) \phi(v) + V(w) \phi(w)~.\end{equation}
Throughout this paper, we assume that $A$ decays exponentially and $V$ is bounded, i.e.\ there exists $\varsigma > 0$ such that
\begin{equation}\label{eq:assum} \sum_{w \in \mathbb Z^\nu} |A(w)| e^{\varsigma \|w\|} < \infty~, \quad  \|V\|_\infty < \infty~, \end{equation}
where $\|w\|$ is the max-norm of $w \in \mathbb Z^\nu$. Unless otherwise stated, all the constants ($C, c$ et cet.) will depend on $\nu$, $\varsigma$, and the values of the quantities in (\ref{eq:assum}),  but not on the other parameters of the problem.

When $H$ is the Hamiltonian of a quantum particle evolving according to Schr\"odinger dynamics, one of the main objects of interest is the quantum transport probabilities
\begin{equation}\label{eq:defP}
P_t(w)  = |e^{itH}(0, w)|^2~, \quad t \geq 0~, \, w \in \mathbb Z^\nu~,
\end{equation}
defined by the Schr\"odinger dynamics of $H$, and the associated transport moments 
\begin{equation}\label{eq:Mp}
M_p(t) = \sum_{w \in \mathbb Z^\nu} P_t(w) \|w\|^p~, \quad t \geq 0~, \, p > 0~.
\end{equation}
The quantity $P_t(w)$ is the probability that a quantum particle starting from the origin is found at $w$ after time $t$, whereas $M_p(t)$ is the $p$-th moment of the position operator. It is known (see (\ref{eq:subbal}) below) that
\begin{equation}\label{eq:subbal-1}
P_t(w) \leq C e^{-c \|w\|}~, \quad \|w \| \geq Ct~,
\end{equation}
and consequently
\begin{equation}\label{eq:subbal-2}
M_p(t) \leq C^p (t^p + (p+\nu)^{p+\nu})~, 
\end{equation}
i.e.\ the motion of the quantum particle is not faster than ballistic. The bounds (\ref{eq:subbal-1}) and (\ref{eq:subbal-2}) give the correct order of growth when $V \equiv 0$, and, more generally, when $V$ is periodic. The Guarneri--Combes--Last bound asserts  that, for each $p$,  $M_p(T)$ grows at least as fast as $T$ to the power $p D/ \nu$, where $D$ is the upper  dimension of the spectral measure (see the work of Last, \cite[Theorem 6.1]{Last}, for the precise formulation, and Landrigan--Powell,  \cite[Theorem 7.3]{LP}, for a recent generalisation). 

On the other extreme, $H$ is said to exhibit dynamical localisation if the family $\{P_t\}_{t \geq 0}$ is tight (precompact in weak topology). A stronger version of localisation is the boundedness of $M_p(t)$ as $t \to \infty$, for some, or even for all, $p > 0$. Many examples of dynamical localisation appear in the framework of random operators with independent, identically distributed potential: for such  operators, dynamical localisation  holds when $\nu  = 1$ and, for $\nu > 1$, when the magnitude of fluctuations of the potential is large enough. The first such result was established by Aizenman \cite{Aiz}, whose work followed extensive research of the spectral properties of such operators (particularly, by Goldsheid--Molchanov--Pastur \cite{GMP} in dimension $\nu=1$ and by Fr\"ohlich--Spencer \cite{FS} and Aizenman--Molchanov \cite{AM} in higher dimension);  see   \cite{AW} for a survey of prior and subsequent developments. 

According to a theorem of Ruelle, Amrein, Georgescu, and Enss (see \cite[\S 2.4]{AW}), dynamical localisation implies that the spectral measure of $H$ at the origin is pure point. The reverse implication is in general not true. It is also not true that low dimensionality of the spectral measures implies any upper bounds on the moments $M_p(t)$: an example of Last \cite{Last} shows that $M_p(t)/t^p$ can tend to zero arbitrarily slowly (as $t \to \infty$) even for operators the spectral measures of which are supported on a set of zero Hausdorff dimension.

Our goal here is to obtain upper bounds on $M_p(t)$ (and on the rate of escape of probability for $P_t$) for operators the spectrum of which is ``thin'' but either dynamical localisation  fails or it remains outside of the scope of the existing technology (see Damanik and Fillman \cite{DF} for a recent survey devoted to  discrete and continual Schr\"odinger operators with thin spectra). 

This circle of questions has been extensively studied for $\nu=1$, and particularly for one-dimensional discrete Schr\"odinger operators (corresponding to $A = \mathbbm 1_{\{\pm 1\}}$ in (\ref{eq:defop})).  A general framework allowing to upper-bound the quantum transport probabilities in terms of finite volume Green functions was developed by Damanik and Tcheremchantsev \cite{DTch07,DTch08}, in the setting of discrete Schr\"odinger operators in dimension $\nu = 1$ (the method  used  in \cite{DTch07} to control the integrals of transfer matrix norms   makes use of an idea applied by Jitomirskaya--Last \cite{JL} in their study of the spectral dimension) . Building on this framework, \cite{DTch07,DTch08} established sub-power-law ($t^{-1} \log M_p(t)  \to 0$) growth of moments for several classes of ergodic one-dimensional operators with positive Lyapunov exponent, valid for every (rather than just almost every) phase. These results have been generalised and strengthened by Jitomirskaya--Mavi \cite{JM} and Han--Jitomirskaya \cite{HJ}. Recently, Jitomirskaya and Powell \cite{JP} sharpened some of these upper bounds to power-logarithmic growth, $M_p(t) \lesssim \log^{c p}  (t + e)$ (here and forth $A \lesssim B$ if $A/B$ is bounded).
All of these results make use of transfer matrices, and are thus specific to one-dimensional Schr\"odinger operators.

Beyond the  discrete Schr\"odinger setting, Jitomirskaya and Liu \cite{JL} developed a method to prove sub-power-law upper bounds on the moments for long-range operators of the form (\ref{eq:defop}) (with decaying kernels $A$ as in (\ref{eq:assum})). However, they still rely on the specifics of one dimension: transport is slowed down by a  barrier separating the origin from infinity. We are not aware of any previous upper bounds on $M_p(t)$ for ergodic Schr\"odinger operators in higher dimension, except for those which are obtained as a consequence of dynamical localisation, and this is the gap that we would like to rectify.  

\medskip \noindent
Similarly to some of these earlier works, the input we require is a certain bound on the Green function in finite volume. One of the main new features of our work is  that we allow an exceptional set of energies of small measure at which there is no control of the Green function: a complex-analytic argument is employed to  show that the energies in this exceptional set do not influence the quantum dynamics on the relevant timescales. The great advantage (from our point of view) is that the assumptions of our main theorem can be verified using fixed-energy analysis. For comparison, the proofs of dynamical localisation typically require simultaneous treatment of all the energies in the spectrum (the arguments of Aizenman \cite{Aiz} for random operators are a notable exception that proves the rule).

We apply our general result to ergodic, and, more specifically, quasiperiodic operators.  In this case the assumptions of the theorem can be deduced using the large deviation estimates for the finite-volume Green functions. We recall that such large deviation estimates form one of the two main building blocks in the strategy, introduced by Bourgain and Goldstein in \cite{BG} and developped in subsequent works (see the monograph of Bourgain \cite{Bourg-book} and references therein, and also the lecture notes of Schlag \cite{Schlag}), for proving Anderson localisation in the quasiperiodic setting. For many classes of operators, such estimates can be proved under arithmetic assumptions on the frequencies. We refer to the recent work of Liu \cite{Liu}, who strengthened, extended, and quantified many of the previous results.

On the other hand, the subsequent arguments leading to locali\-sation are of more delicate nature, and often require non-arithme\-tic or non-explicit assumptions on the frequency \cite{BGS-skew} or even excluding a positive measure of phases \cite{BGS-qp}. As put forth in the cited works \cite{DTch07,DTch08} and especially in \cite{JL,JP}, sub-power-law and power-loga\-rithmic estimates on the transport moments is a moderate relaxation of dynamical localisation that could be valid for all phases (without exclusion of a set of positive or even zero measure). We contribute to this line of research by providing a new strategy, which can be used to recover most of the existing results and prove new ones, lying beyond the range of previous methods. We remark that while the general bounds obtained from the main theorem are valid for almost all phases, an additional argument allows to upgrade them to uniform estimates for all phases. It is known that  dynamical localisation  can not hold uniformly in the phase; in this aspect,  the slow quantum motion  considered in this paper differs significantly from localisation.

In Section~\ref{s:appl}, we consider quasiperiodic operators in arbitrary dimension with potential defined by an irrational shift. The main application is stated as Corollary~\ref{cor:nu-k}; to the best of our knowledge, it provides the first bound on the quantum dynamics associated with quasiperiodic operators in dimension $\nu > 1$ which is valid without exclusion of a positive measure of  phases. An additional feature of this  result is that the assumptions on the frequency are purely arithmetic.
\begin{manualtheorem}{Main application}[simplified version of Corollary~\ref{cor:nu-k}]
Let  $F: \mathbb T^k \to \mathbb R$ be an analytic function not constant on any line segment, and let $\alpha_1, \dots, \alpha_\nu \in \mathbb R^k \setminus \mathbb Q^k$.
For $\coupling > 0$ and $\theta \in \mathbb T^k$, consider an operator  $H_\theta$  acting on $\ell_2(\mathbb Z^\nu)$ as in (\ref{eq:defop})--(\ref{eq:assum}),  where $V_\theta(w) =  \coupling F(\theta + w_1 \alpha_1 + \cdots + w_\nu \alpha_\nu)$.   
Suppose   
\[\begin{cases}
\text{either $k = 1$, $\dist(\alpha w, \mathbb Z) \gtrsim \|w\|^{-\kappa}$ ($w \in \mathbb Z^\nu \setminus \{0\}$) for some $\kappa  \geq \nu$} \\
\text{or $\nu = k = 2$, $\alpha_{1,2} = \alpha_{2,1} = 0$, $\dist(n \alpha_{j,j}, \mathbb Z) \gtrsim |n|^{-\kappa}$ ($n \in \mathbb Z \setminus \{0\}$) for some $1 \leq \kappa < 13/12$.}
\end{cases}\]
Then for $\coupling \geq \coupling_0(\kappa, F)$ the moments $M_p(t)$ grow at most power-logarithmically, in the following strong sense:
\[ \sup_{p > 0} \limsup_{t \to \infty} \frac{\log \sup_{\theta\in \mathbb T^k} \sum_{w \in \mathbb Z^\nu} |e^{i t H_\theta}(0, w)|^2 \|w\|^p}{p  \log \log t} < \infty~.\]
\end{manualtheorem}
In Corollaries~\ref{cor:1-1}, \ref{cor:1-k} we use our method to recover and augment some of the results in dimension $\nu=1$ proved by Damanik--Tcheremchantsev \cite{DTch07}, Jitomirskaya--Liu \cite{JL}, and Jitomirskaya--Powell \cite{JP}.

In Section~\ref{s:appl2}, we discuss a more general framework applicable to other base dynamics. As an illustration, we consider the one-dimensional skew-shift; Corollary~\ref{cor:skew}  generalises a recent result of Jitomirskaya--Powell \cite{JP} to the long-range setting.

As we mentioned already, the applications discussed in Sections~\ref{s:appl} and \ref{s:appl2} rely crucially on the existing large deviation estimates; particularly, we borrow several such estimates from the work of Liu \cite{Liu}.
 
\paragraph{The main result}

To state the main result, introduce the following notation. For a finite volume $\Lambda \subset \mathbb Z^\nu$, denote by $H_\Lambda$ the restriction of $H$ to $\Lambda$ (i.e.\ $H_\Lambda = P_\Lambda H P_\Lambda^*$, where $P_\Lambda:  \ell_2(\mathbb Z^\nu) \to \ell_2(\Lambda)$ is the co\"ordinate projection). For $z$ in the resolvent set of $H$, set  $G_z = (H - z)^{-1}$; similarly, set  $G_{z, \Lambda} = (H_\Lambda - z)^{-1}$. If $\Lambda \ni 0$ and $\epsilon > 0$, denote by
\[ \Res^*(\Lambda; \epsilon) = 
\left\{ E \in \mathbb R\,:\, \exists v \in \Lambda, w \notin \Lambda \, \text{ such that } \, |G_{E, \Lambda}(0, v)| > \epsilon, \, |A(v - w)| > \epsilon \right\}\]
the set of real energies at which the Green function does not decay sufficiently well from the center to the neighbourhood of the boundary. The theorem controls the  escape of probability of $P_t$ in terms of $\epsilon$ and the Lebesgue measure $\mes (\Res^*(\Lambda; \epsilon))$ of this set. The precise formulation   is as follows.

\begin{thm}\label{thm} Let $H$ be as in (\ref{eq:defop})--(\ref{eq:assum}). 
Let $\mathbb Z^\nu \supset \Lambda_1, \cdots, \Lambda_M \ni 0$ be finite volumes, $\Lambda = \bigcup_{m=1}^M \Lambda_m$. Assume that $0 < \epsilon \leq \delta^{8(\nu + 1)M}  \leq 1$ are such that
\begin{equation} \label{eq:assum-thm}
\mes \bigcap_{m=1}^M \Res^*(\Lambda_m; \epsilon) \leq \delta~.
\end{equation}
Then for $w \notin  \Lambda$ and $0 \leq t  \leq \frac{1}{40 M \delta} |\log \epsilon|$
\begin{equation}\label{eq:thm}
P_t(w) \leq C\left[  e^{-c \|w\|} \delta^2 +    |\Lambda|^2  \epsilon^{\frac{1}{5 M}} \right]~.
\end{equation}
\end{thm}
 
\begin{rmk} The arguments leading to Theorem~\ref{thm} are rather flexible, and allow for several possible modifications, of which we mention three:
\begin{enumerate}
\item the sets $\Res^*(\Lambda; \epsilon)$ in the assumptions can be replaced with 
\[ \Res_\delta^*(\Lambda; \epsilon) = 
\left\{ E \in \mathbb R\,:\, \exists v \in \Lambda, w \notin \Lambda \, \text{s.t.} \, |G_{E+i\delta, \Lambda}(0, v)| > \epsilon, \, |A(v - w)| > \epsilon \right\}~,\]
i.e.\ we can slightly complexify the energies;
\item  in the case $M = 1$ (one finite volume), the power of $\epsilon$ in the right-hand side of (\ref{eq:thm}) can be made arbitrarily close to $2$ by adjusting the other parameters (cf.\ the second line of (\ref{sme})); 
\item one can make  the exponent $c$ in the first term of (\ref{eq:thm}) arbitrarily close to $2 \varsigma$, where $\varsigma$ comes from (\ref{eq:assum}).
\end{enumerate}
\end{rmk}

\paragraph{Acknowledgement} We are grateful to Svetlana Jitomirskaya for very helpful correspondence.

\section{Irrational shift}\label{s:appl}

In this section we consider quasiperiodic operators in dimension $\nu \geq 1$ generated by an irrational shift. Let $k \geq 1$; we consider the $k$-dimensional torus $\mathbb T^k$ as a probability space, and denote the  Lebesgue measure on it by $\mathbb P$. Let $\alpha_1, \cdots, \alpha_\nu \in \mathbb T^k \setminus \mathbb Q^k$ be a tuple of frequencies. For $w  = (w_1, \cdots, w_\nu) \in \mathbb Z^\nu$, let 
\[ T^v: \mathbb T^k \to \mathbb T^k~, \quad T^w (\theta) = \theta + w_1 \alpha_1 + \cdots + w_\nu \alpha_\nu = \theta + \alpha w~.\]
Finally, let $F: \mathbb T^k \to \mathbb R$ be a hull function, which we always assume to be non-constant. Given this data, define a family of operators $H_\theta$, $\theta \in \mathbb T^k$, acting on $\ell_2(\mathbb Z^\nu)$ via
\begin{equation}\label{eq:shift}
H_\theta \phi = A*\phi + V_\theta \phi~, \quad V_\theta(w) = F(T^w \theta)~. 
\end{equation}
To emphasise the dependence on $\theta$ we denote the corresponding quantum probabilities by $P_{t, \theta}(w)$, and the moments by $M_p(t, \theta)$; sometimes we also emphasise the dependence of $\Res^*(\Lambda, \delta) = \Res^*(\Lambda, \delta; \theta)$ on $\theta$. Also set $P_{t}(w) = \sup_\theta P_{t, \theta}(w)$ and $M_p(t) =\sup_\theta M_p(t, \theta)$.  

The theorem implies the following corollary which controls the quantum probabilities uniformly in $\theta$. In the formulation,   we gave preference to shorter formul{\ae} over the sharpest possible bounds. 

\begin{cor}\label{cor:shift} Assume that $F$ is Lipschitz-continuous, and that $M \geq 1$, $N \geq 1$, $0\in \Lambda_1, \cdots, \Lambda_m \subset [-N,N]^\nu$, and $0 < \delta \leq 1$ are such that
\begin{equation}\label{eq:LD-G}
\forall E \in \mathbb R \,\,\,  \mathbb P \left\{ E \in\bigcap_{m=1}^M \Res^*(\Lambda_m, \delta^{8(\nu+1)M}; \theta) \right\} \leq \delta^{2k+1}~.
\end{equation} 
Then we have (for $C>0$ depending on $\nu$, $k$, (\ref{eq:assum}) and the Lipschitz constant of $F$):
\begin{equation}
\forall  \|w\| > N~, \, t \leq \frac{1}{\delta}: \quad  P_{t}(w) \leq C  N^{2\nu} \delta~.
\end{equation}
\end{cor}

\begin{proof}
Let us call $\theta \in \mathbb T^k$ typical if $\mes \bigcap_{m=1}^M \Res^*(\Lambda_m, \delta^{8(\nu+1)M}; \theta) \leq\delta$. By the Fubini theorem and the Chebyshev inequality,
\[ \mathbb P(\text{$\theta$ is atypical}) 
\leq \frac1\delta\, \mathbb E \mes \bigcap_{m=1}^M \Res^*(\Lambda_m, \delta^{8(\nu+1)M}; \theta) \leq \frac1\delta \, C_1\delta^{2k+1} = C_1 \delta^{2k}~,\]
with $C_1$ depending on the operator norm. In particular, the set of atypical $\theta$-s can not contain a disk of radius $C_2 \delta^2$ (for $C_2$ such that the volume of a $k$-dimensional ball  is equal to $C_1 / C_2^k$).
If $\theta$ is typical, Theorem~\ref{thm} (with $\epsilon = \delta^{8(\nu+1)M}$) implies that 
\begin{equation}\label{eq:cor-shift-need} P_{t,\theta}(w) \leq C_3 N^{2\nu} \delta^{\frac{8(\nu+1)}{5}} \leq C_3 N^{2\nu} \delta~, \quad \|w\| > N~, \, t \leq \frac{1}{\delta}~.\end{equation}
If $\theta$ is atypical, we can find a typical $\theta'$ at distance at most $C_2 \delta^2$ from $\theta$. For this $\theta'$ we have $\|H_\theta - H_{\theta'} \| \leq C_4\delta^2$ (where $C_4/C_2$ is the Lipschitz constant of $F$), hence (by Duhamel's formula) $\|e^{it H_\theta} - e^{itH_{\theta'}} \| \leq C_4\delta^2t$. Then for $t \leq \frac{1}{\delta}$
\[ P_{t,\theta}(w) \leq P_{t, \theta'}(w) + 2 C_4 \delta^2 t \leq P_{t, \theta'}(w) + 2 C_4 \delta~, \]
and thus (\ref{eq:cor-shift-need}) holds for $\theta$ with an adjusted value of the constant ($C_1 + 2C_4$). 
\end{proof}

Now we list some of the possible applications of the corollary.  The one-dimensional case ($\nu=1$) was recently studied by Jitomirskaya--Liu \cite{JL} and Jitomirskaya--Powell \cite{JP}. What we add to these works is uniformity of the estimates in the phase, and, in the long-range case, an improvement from sub-power-law to power-logarithmic pounds. 

The results in higher dimension ($\nu \geq 2$) are new, and seem  to be the first upper bounds on the quantum transport that are valid for all $\theta$ (in fact, previous works require exclusion of a set of $\theta$ of positive measure). 

\medskip \noindent
As usual, we say that $(\alpha_1, \cdots, \alpha_\nu) \in \mathbb T^{k\nu}$ satisfies the Diophantine condition DC($\kappa,\tau$) if for $w \in \mathbb Z^\nu \setminus \{0\}$ the distance beween $T^w (0)$ and $0$ is lower-bounded by $\tau \|w\|^{-\kappa}$. We also say that $(\alpha_1, \cdots, \alpha_k)$ satisfies the weak Diophantine condition WDC($\kappa$) if the same quantity is lower-bounded by $c_\zeta e^{-\zeta \|w\|^\kappa}$ for an arbitrarily small $\zeta> 0$, and  the strong Diophantine condition SDC($\kappa$) if it is lower-bounded by $c \|w\|^{-1} \log^{-\kappa} (\|w\|+e)$.

We also record the following elementary lemma, which will allow us to translate the estimates on $P_t$ into estimates on $M_p$.
\begin{lemma}\label{l:elem} Fix $\varrho>0$. Let $P_t$, $t \geq 0$, be probability distributions on $\mathbb Z^\nu$ satisfying (\ref{eq:subbal-1}), and let $M_p(t) = \sum_{w \in \mathbb Z^\nu} P_t(w) \|w\|^p$. Assume that $N_1 < N_2 < \cdots$ is an increasing sequence of natural numbers such that for each $N=N_j$
\begin{equation}\label{eq:l-elem-assum} P_t(w) \leq e^{-r N^\varrho} \quad \text{for} \quad 0 < t \leq e^{r N^\varrho}~, \,\, \|w \| \geq N~.\end{equation}
Then for any $p > 0$ we have the following:
\begin{enumerate}
\item If $N_{j+1} \lesssim N_j^\kappa$ for some $\kappa \geq 1$, then  $M_p(t)\lesssim \log^{\kappa p / \varrho}(t+e)$;
\item If $N_j^{-\varrho} \log  N_{j+1} \to 0$, then $\log M_p(t) / \log t \to 0$;
\item without  assumptions on $N_j$, there exists a sequence $t_j \to \infty$ such that $M_p(t_j) \lesssim \log^{p/\varrho}(t + e)$.
\end{enumerate}
The implicit constants in the conclusions depend only on $\nu$ and the implicit constants in (\ref{eq:subbal-1}) and the assumptions.
\end{lemma}
\begin{proof}
Let $j(t) = \min\{j \, : \, t  \leq e^{aN^\varrho}\}$. If $N_{j+1} \lesssim N_j^\kappa$, then $N_{j(t)} \lesssim \log^{\kappa/\varrho} t$, whereas if $\frac{1}{N_j^\varrho} \log  N_{j+1}$ tends to $0$, then so does  $\log N_{j(t)} / \log t$. In either case,
\[ \begin{split}
M_p(t) &= \sum_{w \in \mathbb Z^\nu} P_t(w) \|w\|^p = \sum_{\| w \| \leq N_{j(t)}} + \sum_{j=j(t)}^\infty \sum_{N_j < \| w\| \leq N_{j+1}}\\
&\leq  N_{j(t)}^p  + \sum_{j=1}^\infty (2N_{j+1}+1)^{\nu + p}  e^{-r N_j^\varrho}~,
\end{split}\] 
where the sum is convergent. This proves the first and the second items. For the third one, let $t_j = e^{\frac{r}{2p+1} N_j^\varrho}$, then
\[ M_p(t_j) = \sum_{w\in \mathbb Z^\nu} P_t(w) \|w\|^p = \sum_{\| w\| \leq N_{j}}  + \sum_{N_j < \| w \| \leq C t_j} + \sum_{\|w \|> C t_j}~. \]
The first sum is bounded by $N_j^p \lesssim \log^{p/\varrho} t_j$, the second one is bounded by $(Ct_j)^p e^{-r N_j^\varrho} \lesssim 1$ by assumption, and the last one is $\lesssim 1$ by (\ref{eq:subbal-1}). 
\end{proof}

\paragraph{{Discrete Schr\"odinger operators in dimension  $\nu  = 1$ with $k = 1$ frequency}}

In this case the strongest results can be obtained. Denote by $\Phi_N(E; \theta)$ the $N$-step transfer matrix at energy $E$, and by $\gamma(E)$ -- the Lyapunov exponent. As an input to our general results, we shall need  a large deviation estimate for all $E \in \mathbb R$:
\begin{equation}\label{eq:LD-TM} \forall   \zeta > 0: \,\, \mathbb P \left\{ \left| \log  \|\Phi_{N, \theta}(E)\| - \gamma(E) N \right| \geq \zeta N   \right\} \leq e^{-r(\zeta) N^{\varrho}}~.\end{equation} 
If such an estimate holds, it implies (\ref{eq:LD-G}) with $\delta = e^{-r N^{\varrho}}$ (for sufficiently small $r>0$) and $M=4$; as usual, the volumes $\Lambda_m$ are $[-L, R]$, $[-L+1, R]$, $[-L, R-1]$, and $[-L+1, R-1]$, where, say, $L = \lfloor N/2 \rfloor$, and $R = L + N$.

A relatively low-technology estimate of  Bourgain--Jitomirskaya \cite[Lemma 4]{BJ} (requiring averaging shifts of subharmonic functions but not multiscale arguments, and valid without any arithmetic assumptions on $\alpha$) implies the following: if   $F$ is analytic and  there exists a rational $p/q$ with $|\alpha - p/q| < 1/q^2$, then (\ref{eq:LD-TM}) holds with $\varrho = 1$ for $K q \leq N\leq 2 K  q$ for some $K > 1$ depending on $\coupling F$ and $\zeta$). 
 
 S.~Klein \cite{Kl1,Kl2} studied the case when $F$ lies in a Gevrey class  $G_s(\mathbb T)$, $s \geq 1$ (i.e.\ the $j$-th derivative of $F$ is uniformly bounded by $(C_Fj)^{sj}$). Combining  the approximation argument from \cite[Sections 2--4]{Kl1} with the method of \cite[Lemma 4]{BJ}, one obtains that if $1 \leq s < 2$ and  there exists a rational $p/q$ with $|\alpha - p/q| < 1/q^2$, then (\ref{eq:LD-TM}) holds with $\varrho  = (2-s)/s$ for $K q^{s} \leq N \leq 2K q^{s}$ (cf.\ \cite[Proposition~5]{CGYZ}, where this fact is stated with implicit exponents). To summarise:
 
 \begin{lemma}[\cite{BJ,Kl1}]\label{l:BJKl} Let $1 \leq s < 2$, and assume that $F \in G_s(\mathbb T)$, $\coupling > 0$. For any $\zeta > 0$ there exist $r$ and $K$ depending on $\coupling F$ and $\zeta$, such that (\ref{eq:LD-TM})  holds for any $E \in \mathbb R$ whenever $K q^{s} \leq N \leq 2K q^{s}$ and $|\alpha - p/q| < 1/q^2$.
 \end{lemma}

We postpone the discussion of the case $s \geq 2$ (also based on the work of Klein, \cite{Kl1,Kl2}) to the next paragraph.

\begin{cor}\label{cor:1-1} Consider a discrete one-dimensional Schr\"odinger operators defined by an irrational shift on $\mathbb T$ with a non-constant hull function $F \in G_s(\mathbb T)$, $1 \leq s < 2$. Assume that the Lyapunov exponent is positive at all energies. Then $M_p(t)$, the supremum of the $p$-th moment $M_p(t,\theta)$ over all $\theta \in \mathbb T^k$, boasts the following estimates.
\begin{enumerate}
\item If $\alpha$ satisfies DC($\kappa,\tau$), then $M_p(t) \lesssim \log^{p\kappa s /(2-s)}(t +e)$ for any $p > 0$.
\item If $\alpha$ satisfies WDC($2-s$), then $\frac1t \log M_p(t) \to 0$ as $t \to \infty$ for any $p > 0$.
\item For arbitrary $\alpha$, there exists a sequence $t_j \to +\infty$ such that $M_p(t_j) \lesssim \log^{ s / (2-s)}(t_j +e)$.
\end{enumerate}
\end{cor}

\begin{rmk} \hfill
\begin{enumerate}
\item 
A version of the second item of the corollary for the case when $F$ is a trigonometric polynomial has been proved by Damanik and Tcheremchantsev \cite[Theorem 5]{DTch07}. Variants of the first and third items (with different exponents) have been  recently proved by Jitomirskaya and Powell \cite[Theorem 1.1]{JP}. Uniformity in the phase is not discussed in these works. 
\item The dependence on $\kappa$ in the first item may be removed with the help of  more refined large deviation estimates such as   \cite[Theorem 4.1]{Kl1}.
\end{enumerate}
\end{rmk}

\begin{rmk} Last showed \cite{Last} that the conclusions of the first two items  of Corollary~\ref{cor:1-1} can not hold for well-approximated frequencies, and in fact one can not improve on the ballistic uniform bound: for any $\zeta(t) \searrow 0$  there exist frequencies $\alpha$ for which there are  arbitrarily large times $t$  such that $M_p(t) \geq t^p  \zeta(t)$. 
\end{rmk}

\begin{proof}[Proof of Corollary~\ref{cor:1-1}]
Let $q_j$ be the sequence of denominators of the convergents of $\alpha$. If $\alpha$ satisfies DC($\kappa,\tau$), then $q_{j+1} \lesssim q_j^\kappa$, whereas if  $\alpha$ satisfies WDC($2-s$), then $\log q_{j+1} / q_j^{2-s} \to 0$. For each $j$, let $N_j = \lceil K q^{s} \rceil$.  Now  Lemma~\ref{l:BJKl} (with, say, $\zeta = \frac{1}{10} \min_{E} \gamma(E)$) allows to verify the assumptions of  Corollary~\ref{cor:shift} with $\delta_j =e^{- r N_j^{\varrho}}$, with $\varrho = (2-s)/s$ and sufficiently small $r > 0$. The conclusion of Corollary~\ref{cor:shift} implies that these $N_j$ satisfy the assumption (\ref{eq:l-elem-assum}) of the Lemma~\ref{l:elem}.  In the Diophantine case, $N_{j+1} \lesssim N_j ^\kappa$, while in the weakly Diophantine case, $\log N_{j+1} / N_j^{\frac{2-s}{s}} \to 0$. Thus the three items of Corollary~\ref{cor:1-1} follow from the three items of Lemma~\ref{l:elem}.
\end{proof}

\paragraph{{Schr\"odinger and long-range operators in dimension $\nu =  1$ with $k \geq 1$ frequencies}}

As before, we rely on large deviation estimates of the form (\ref{eq:LD-TM}). For analytic hull functions such estimates go back to the work of Bourgain, Goldstein, and Schlag \cite{BG,BGS-qp,Bourg-book,GS}; we rely on the bounds of Liu \cite{Liu}, who streamalined, sharpened and quantified the previous results (we refer to his work for further references).  The extension to Gevrey hull functions are due to Klein \cite{Kl1,Kl2}, whose estimates are somewhat less explict.
The results from these works are summarised in the following Table~\ref{tab:Liu}.

Throughout the table, we assume that that $\alpha$ is assumed to satisfy the Diophantine condition DC($\kappa, \tau$). The number in the first column is the number of frequencies $k$. In the second column, we specify whether $A = \mathbbm 1_{\{\pm 1\}}$ (namely, we have a discrete Schr\"odinger operator) or an arbitrary kernel satisfying (\ref{eq:assum}) is allowed.   The third column indicates whether the result holds  whenever the Lyapunov exponent is positive, or only for sufficiently large coupling constant (depending on all the parameters of the problem, including the frequency $\alpha$). The fourth column shows the Gevrey class of  $F$. If $s > 1$, we assume  that all the derivatives of $F$ of all orders do not vanish simultaneously at any point (for $s = 1$, this follows from analyticity). 
The fifth column provides the value of $\varrho$; expressions such as $a + 0$ should be interpreted as $a + \zeta$ for an arbitrarily small $\zeta > 0$.  

\begin{table}[h!] 
  \begin{center}
    \caption{Large deviation estimates for the Green functions of one-dimensional operators ($\nu = 1$), analytic  hull function $F$, and Diophantine $\alpha$}
    \label{tab:Liu}
\vspace{2mm}
    \begin{tabular}{||c|c|c|c||c|| c||}  
      \hline
      \textbf{frequencies} & \textbf{kernel ($A$)} &   \textbf{assumptions} & $G_s$& $\varrho$ & \textbf{reference}\\
      \hline\hline
	$k\geq 1$  & $\mathbbm 1_{\{\pm 1\}}$ &   LE $>0$ & $s = 1$ & $\frac{1}{k^{3} \kappa^{2}+0}$ & \cite[Thm. 3.1]{Liu}\\ \hline
	$k \geq 1$ & arbitrary &      $\coupling \geq \coupling_0$ &$s=1$& $\frac{1}{k^{3} \kappa^{2}+0}$ & \cite[Thm. 3.8]{Liu} \\\hline
	$k =1$ & $\mathbbm 1_{\{\pm 1\}}$ & $\coupling \geq \coupling_0$  &  $s \geq 1$ & implicit $\varrho$ &\cite[Thm. 5.1]{Kl1} \\\hline
	$k =1$ & $\mathbbm 1_{\{\pm 1\}}$ & $\coupling \geq \coupling_0$  &  $s \geq 1$ & implicit $\varrho$ & \cite[Thm. 6.1]{Kl2} \\\hline
    \end{tabular}
  \end{center}
\end{table}

\begin{cor}\label{cor:1-k} Consider a one-dimensional    operator $H_\theta$ with kernel $A$ satisfying (\ref{eq:assum}) and potential defined by an irrational shift on $\mathbb T^k$ with a non-constant hull function $F \in G_s(\mathbb T^k)$. Under the conditions given in  each of the rows in Tables~\ref{tab:Liu},  one has  $M_p(t) \lesssim \log^{p/\varrho} (t + e)$ for any $p > 0$.  
\end{cor}

\begin{rmk} The results corresponding to the first, third and fourth row are variants of the recent work of Jitomirskaya--Powell \cite{JP}. For the second row, a sub-power-law estimate $t^{-1} \log M_p(t) \to 0$ was obtained by Jitomirskaya and Liu \cite{JL}.
\end{rmk}

\begin{proof}[Proof of Corollary~\ref{cor:1-k}]
Apply Corollary~\ref{cor:shift} with the large deviation estimates in Table~\ref{tab:Liu} and then Lemma~\ref{l:elem}.
\end{proof}

\paragraph{{Schr\"odinger and long-range operators in dimension $\nu > 1$}}

Table~\ref{tab:Liu2}, also based on the work of Liu \cite{Liu}, summarises similar estimates in arbitrary dimension $\nu \geq 1$.   These results  build on earlier works of Bourgain  (see \cite{Bourg}) and  Bourgain--Kachkovskiy \cite{BK}; again we refer to \cite{Liu} for references.

 Now we always assume that the coupling constant is large enough and that $F$ is analytic. In the second row, we assume
\begin{equation}\label{eq:3.20}
\begin{split} 
&\text{$F$ is non-constant on any line segment,} \\
&\text{$\alpha_{1,2} = \alpha_{2,1} = 0$, $\alpha_{1,1}, \alpha_{2,2}$ satisfy DC($\kappa, \tau$) for some $1 \leq \kappa < 13/12$.}
\end{split}
\end{equation}

\begin{table}[h!] 
  \begin{center}
    \caption{Large deviation estimates for the Green functions of operators in higher dimension, with analytic  hull function $F$, and Diophantine $\alpha$}
    \label{tab:Liu2}
\vspace{2mm}
    \begin{tabular}{||c|c|c|c|c||c|| c}  
	\hline
     \textbf{dimension} & \textbf{frequencies} & \textbf{kernel ($A$)} & \textbf{assumptions} & $\varrho$ & \textbf{reference}\\
      \hline\hline
	$\nu \geq 1$ & $k=1$ & arbitrary   & $\coupling \geq \coupling_0$  & $1-0$  & \cite[Thm. 3.11]{Liu} \\
      $\nu = 2$ & $k=2$ & arbitrary  & $\coupling \geq \coupling_0$  and (\ref{eq:3.20}) & $(3.25 - 3\kappa-0)^2$ & \cite[Thm.\ 3.20]{Liu}\\
	\hline
    \end{tabular}
  \end{center}
\end{table}

\begin{cor}\label{cor:nu-k}
Consider a $\nu$-dimensional operator $H_\theta$ with kernel $A$ satisfying (\ref{eq:assum}) and potential defined by an irrational shift on $\mathbb T^k$ with a non-constant analytic hull function $F \in G_1(\mathbb T^k)$. Under the conditions given in  each of the rows in Tables~\ref{tab:Liu2},  one has   $M_p(t) \lesssim \log^{p/\varrho} (t + e)$ for any $p > 0$. 
\end{cor}

As we have mentioned earlier, all the previous  bounds on the quantum transport in dimension $\nu > 1$ that we are aware of required exclusion of a positive measure of phases and non-arithmetic conditions on the frequencies.

\begin{rmk} The cases $\nu = k \geq 3$ and arbitrary $\nu, k$ have been considered in the works of Bourgain \cite{Bourg} and Jitomirskaya--Liu--Shi \cite{JLS}, respectively. In these works Anderson localisation is proved for large coupling outside a set of $\theta$-s of small positive measure. It is plausible that the large deviation estimates proved in these works could be used to establish  upper bounds on the moments $M_p(t)$.
\end{rmk}

\begin{cor}\label{cor:nu-k-haus}
In the setting of Corollary~\ref{cor:nu-k}, the spectral measures of $H_\theta$ are supported on a set of zero Hausdorff dimension  and even of zero power-logarithmic Hausdorff measure, for every $\theta \in \mathbb T^k$.
\end{cor}

\begin{proof} Follows from Corollary~\ref{cor:nu-k} combined with the Guarneri--Combes--Last bound \cite[Theorem 6.1]{Last} and its generalisation \cite[Theorem 7.3]{LP}. 
\end{proof}

\begin{rmk}
The conclusion of Corollay~\ref{cor:nu-k-haus} is not surprising for Schr\"odinger operators in dimension $\nu = 1$: according to a  result of Simon \cite{Simon}, the mere  positivity of the Lyapunov exponent implies that the spectral measures are supported on a set of vanishing logarithmic capacity and {\em a fortiori} vanishing power-logarithmic Hausdorff measure (note that in the uniquely ergodic case this holds for all, rather than just almost all, phases $\theta$). On the other hand, we are not aware of any general results of this form in higher dimension, and Corrolary~\ref{cor:nu-k-haus} may provide an interesting addition to the collection of operators with thin spectra (see \cite{DF} for other examples).
\end{rmk}

\begin{proof}[Proof of Corollary~\ref{cor:nu-k}]
Apply Corollary~\ref{cor:shift} with the large deviation estimates in Table~\ref{tab:Liu2} and then Lemma~\ref{l:elem}.
\end{proof}


\section{Other dynamical systems}\label{s:appl2}

Let us first generalise Corollary~\ref{cor:shift}. Let $(\Theta, \mathcal B, \mathbb P, \dist)$ be a metric probability space, and let $T$ be an ergodic action of $\mathbb Z^\nu$ on $\Theta$. Given a hull function $F: \Theta \to \mathbb R$, define a family of operators $H_\theta$, $\theta \in  \Theta$, acting on $\ell_2(\mathbb Z^\nu)$ via
\begin{equation}\label{eq:erg}
H_\theta \phi = A*\phi + V_\theta \phi~, \quad V_\theta(w) = F(T^w \theta)~. 
\end{equation}
The notation $P_{t, \theta}(w)$, $P_t(w)$, $M_p(t, \theta)$, and $M_p(\theta)$ will be used as in the previous section. 

\begin{cor}\label{cor:erg} In the setting described above, assume that  $F$ is Lipschitz-continuous, $k \in \mathbb R_+$ is such that $\Theta$ has finite $k$-dimensional upper Minkowski content, and let $\chi \geq 1$ be such that 
\begin{equation}\label{eq:distort}
\dist(T^w \theta, T^w \theta') \lesssim (|w\|+1)^\chi \dist (\theta, \theta')~.
\end{equation} 
Let  $M \geq 1$, $N \geq 1$, $0\in \Lambda_1, \cdots, \Lambda_m \subset [-N,N]^\nu$, and $0 < \delta \leq 1$ be such that
\begin{equation}\label{eq:LD-G'}
\forall E \in \mathbb R \,\, \mathbb P \left\{ E \in\bigcap_{m=1}^M \Res^*(\Lambda_m, \delta^{8(\nu+1)M}; \theta) \right\} \leq \delta^{(2+\chi)k+1}~.
\end{equation} 
Then we have (for $C>0$ depending on $\nu$, (\ref{eq:assum}), the Minkowski content of $\Theta$, $\chi$, and the Lipschitz constant of $F$):
\begin{equation}
\forall  \|w\| > N~, \, t \leq \frac{1}{\delta}: \quad  P_{t}(w) \leq C  \left[ N^{2\nu} \delta + e^{-c \|w\|}\right]~.
\end{equation}
\end{cor}

\begin{proof}
Similarly to Corollary~\ref{cor:shift}, let  us call $\theta \in \mathbb T^k$ typical if $\mes \bigcap_{m} \Res^*(\Lambda_m, \delta^{8(\nu+1)M}; \theta) \leq\delta$, so that  for typical theta the assertion follows directly from the theorem, and
\[ \mathbb P(\text{$\theta$ is atypical}) \leq C_1\delta^{(2+\chi)k}~. \]
If $\theta$ is atypical, we can find a typical $\theta'$ at distance at most $C_2 \delta^{(2+\chi)}$ from $\theta$. Consider the restrictions $H_{\theta, [-C_3 t, C_3 t]^\nu}$ and $H_{\theta', [-C_3 t, C_3 t]^\nu}$ of $H_\theta$ and $H_{\theta'}$ to finite volumes, where $C_3$ is assumed to be sufficiently large; then by (\ref{eq:distort})  
\[ \| H_{\theta, [-C_3 t, C_3 t]^\nu} - H_{\theta', [-C_3 t, C_3 t]^\nu} \| \leq C_4 t^{\chi} \delta^{(2+\chi)}, \,
 \| e^{it H_{\theta, [-C_3 t, C_3 t]^\nu}} - e^{it H_{\theta', [-C_3 t, C_3 t]^\nu}} \| \leq C_4 t^{\chi+1} \delta^{(2+\chi)} \leq C_4 \delta~. \]
It is not hard to see (and we prove this at the end of Section~\ref{s:proof} below) that 
\begin{equation}\label{eq:subbal3} | e^{it H_{\theta}}(0, w)- e^{it H_{\theta, [-C_3 t, C_3 t]^\nu}} (0, w) | \leq  C e^{-c\|w\|} \end{equation}
(and similarly for $\theta'$), whence  
\[ P_{t,\theta}(w) \leq P_{t, \theta'}(w) + C_4 \delta  + e^{-c\|w\|}~. \qedhere\]
\end{proof}

\noindent As an application, consider the skew-shift: $\nu = 1$, and 
\[ T (\theta_1, \cdots, \theta_k) = (\theta_1 + \alpha, \theta_2 + \theta_1, \cdots, \theta_k + \theta_{k-1})~. \]
where $\alpha \in \mathbb T \setminus \mathbb Q$. In this case, $\chi = k -1 < \infty$.
\begin{cor}\label{cor:skew}
Consider a one-dimensional operator $H_\theta$ with kernel $A$ satisfying (\ref{eq:assum}) and potential defined by the skew-shift on $\mathbb T^k$ with a non-constant  hull function $F \in G_s(\mathbb T^k)$. Suppose that  $\alpha$ satisfies DC($\kappa, \tau$) (if $s=1$) or SDC($2$) (if $s > 1$). Also assume that the derivatives of $F$ of all orders do not vanish at any point (for $s=1$, this holds automatically).

Then for any $0 < \varrho <\varrho_0(k, \kappa)$ and $\coupling$  sufficiently large (depending on $F$ and $\alpha$),   one has  $M_p(t) \lesssim \log^{p/\varrho} (t + e)$ for any $p>0$. In the case $s = 1$ (analytic $F$), one can choose $\varrho_0 =  (4^{k-1}k^3 \kappa^{2})^{-1}$.
\end{cor}

\begin{rmk} The Schr\"odinger case is due to Jitomirskaya and Powell \cite{JP}.  
\end{rmk}

\begin{proof}[Proof of Corollary~\ref{cor:skew}] Apply Corollary~\ref{cor:erg} with $\delta = \exp(-r N^\varrho)$, invoking the large deviation estimates from the works of Bourgain--Goldstein--Schlag \cite{BGS-skew} and Klein \cite{Kl2} to verify the assumption (\ref{eq:LD-G'}) in the cases $s = 1$ and $s > 1$, respectively. The explicit estimate on $\varrho_0$ for $s = 1$ is proved by Liu in \cite[Theorem 3.14]{Liu}. Then apply Lemma~\ref{l:elem}.
\end{proof}


\section{Proof of the main theorem}\label{s:proof}

In the proof, we shall repeatedly use the Combes--Thomas estimate  (stated in sufficiently general form e.g.\ in \cite{Aiz} or in \cite{AW}):
\begin{equation}\label{eq:ct}
|G_z(u, v)| \leq C \delta^{-1} \exp(-c \delta \|u-v\|)~, \quad |G_{z,\Lambda}(u, v)| \leq C \delta^{-1} \exp(-c \delta \|u-v\|)~,  
\end{equation}
where $\delta$ is the minimum between the distance from $z$ to the spectrum and $1$. 
We recall  that (\ref{eq:ct}) implies the ballistic bound 
\begin{equation}\label{eq:subbal}
|e^{itH}(0, w)| \leq C e^{-c \|w\|}~, \quad \|w \| \geq Ct~, 
\end{equation}
using the contour integral representation
\begin{equation}\label{eq:intrepr} e^{itH}(0, w) = - \frac{1}{2\pi i} \oint_{\mathcal C}  e^{i tz} G_z(0,w) dz~.\end{equation}
The contour should encircle the spectrum counterclockwise; for (\ref{eq:subbal}), one chooses  it to be the boundary of the rectangle $|\Im z| \leq 1$, $|\Re z| \leq \|H\|+1$. 
Similarly, the proof of Theorem~\ref{thm} will require control of the Green function on an appropriate contour. This will be achieved using the following main lemma.

\begin{lemma}\label{l:main} Consider a finite collection of functions
\[ u_{m,v}(z) = \sum_{j=1}^{d(m,v)} \frac{a_{j,m,v}}{z - b_{j,m,v}}~, \quad 1 \leq m \leq M~, \,    v \in \mathcal V \]
 such that  all $a_{j,m,v}, b_{j,m,v}$ are real, and let $\epsilon, \delta \in (0, 1]$. If 
\begin{eqnarray}\label{eq:mass1}  
&&\sum_{j=1}^{d(m,v)}|a_{j,m,v}| \leq 1 \quad \text{for every $m, v$;}\\
\label{eq:smae}
 &&\mes \left\{ x \in \mathbb R \, : \, \min_{1 \leq m \leq M} \max_{v \in \mathcal  V} |u_{m,v}(x)| > \epsilon\right\} \leq \delta~,
\end{eqnarray}
then for $z = x + i y \in \mathbb C$ one has
\begin{equation}\label{sme}
\min_{1 \leq m \leq M} \max_{v \in \mathcal  V} |u_{m,v}(x)| \leq  
\begin{cases}
\frac{4 \,  \epsilon^{\frac{1}{2M}}}{|y|}~, & |y| \geq \frac{2}{\pi}\delta, \quad \text{for any $M \geq 1$}\\
 \frac{4\,  \epsilon^{1 - \frac{\delta}{\pi |y|}}}{|y|}~, &|y| > \frac{1}{\pi} \delta, \quad \text{for $M = 1$}~. 
 \end{cases}
\end{equation}
\end{lemma}

\begin{proof}
For $z = x+iy$ with, say, $y > 0$,  let 
\[ d\mu_z(t) = \frac{y}{\pi} \frac{dt}{(t-x)^2 + y^2} \]
be the Poisson measure. For $y \geq \frac2\pi \delta$
\[ \mu_z \left\{ x \in \mathbb R \, : \, \min_{1 \leq m \leq M} \max_{v \in \mathcal V} |u_{m,v}(x)| > \epsilon\right\} \leq \frac{\delta}{\pi y} \leq \frac12~,  \]
hence there exists $1 \leq m \leq M$ such that for any $v \in \mathcal V$
\begin{equation}\label{eq:estM}  \mu_z \left\{ x \in \mathbb R \, : \,   |u_{m,v}(x)| \leq  \epsilon\right\}  \geq \frac{1}{2M}~. 
\end{equation}
In the special case $M=1$, we have the better bound 
\begin{equation}\label{eq:est1}  \mu_z \left\{ x \in \mathbb R \, : \,   |u_{1,v}(x)| \leq  \epsilon\right\} \geq 1 - \frac{\delta}{\pi y}~. 
\end{equation}
For arbitrary $M \geq 1$, we have for $m$ from (\ref{eq:estM}) and all $v \in \mathcal  V$ 
\[\log |u_{m,v}(z)| \leq \int \log |u_{m,v}(t)|  d\mu_z(t) \leq \frac{1}{2M} \log \epsilon + \int \log_+ |u_{m,v}(t)| d\mu_z(t)~.\]
By Boole's identity (see \cite{Aiz,AW} and references therein) and the assumption (\ref{eq:mass1}), we have for any  $r\in \mathbb R $:
\[ \mu_z  \left\{ t \, : \, |u_{m,v}(t)|  \geq e^{r} \right\} \leq  \frac{1}{\pi y} \mes \left\{ t \, : \, |u_{m,v}(t)|  \geq e^{r} \right\} \leq \frac{4}{\pi y} \,  e^{-r}~, \]
hence 
\[ \begin{split} \int \log_+ |u_{m,v}(t)| d\mu_z(t) 
&= \int_0^\infty \mu_z \left\{ t \, : \, |u_{m,v}(t)|  \geq e^{r} \right\} dr \\
&= \int_0^{\log \frac{4}{\pi y}}  + \int_{\log \frac{4}{\pi y}}^\infty \leq \log \frac{4}{\pi y} + 1 = \log \frac{4e}{\pi y} \leq \log \frac{4}{y}~.
\end{split}\]
Finally, we have for the chosen $m$ and all $v \in \mathcal V$: 
\[ \log |u_{m,v}(z)|  \leq \frac{1}{2M} \log \epsilon  + \log \frac{4}{y} \leq  \log  \frac{4\, \epsilon^{\frac1{2M}}}{y}~,\] 
as claimed. For $M= 1$ we use (\ref{eq:est1}) in place of (\ref{eq:estM}) and obtain
\[ \log |u_{m,v}(z)|  \leq (1 - \frac{\delta}{\pi y}) \log \epsilon  + \log \frac{4}{y} =  \log  \frac{4  \epsilon^{1 - \frac{\delta}{\pi y}}}{y}~.\qedhere\] 
\end{proof}

\noindent Now we proceed to the proof of the theorem.  Let  
\[ \mathcal V = \left\{ v \in \Lambda \, : \, \exists w \notin \Lambda, \, |A(v-w)| > \epsilon \right\} \]
be the collection of vertices close to the boundary of $\Lambda$. For each $m \in \{ 1, \cdots, M\}$ and $v \in \mathcal V$, set
\[ u_{m,v}(z) = \begin{cases}
G_{z,\Lambda_m}(0, v)~, &v \in \Lambda_m \\
0~.
\end{cases} \]
These functions satisfy the assumptions of Lemma~\ref{l:main}: (\ref{eq:mass1}) follows from the spectral representation
\[ G_{z, \Lambda_m}(0, v) = \sum_{j=1}^{\Lambda_m} \frac{\psi_j(0) \psi_j(v)}{z -\lambda_j}~, \quad
\sum_{j} |\psi_j(0) \psi_j(v)| \leq \sqrt{\sum_{j} |\psi_j(0)|^2} \sqrt{\sum_{j} | \psi_j(v)|^2 } = 1~, \]
whereas (\ref{eq:smae}) follows from the assumption (\ref{eq:assum-thm}).
Thus for $z = E+iy$ with, say, $|y| \geq \delta$, one has the following: for any $E \in \mathbb R$ there exists $m \in \{1, \cdots, M\}$ such that for all $v \in \mathcal V \cap \Lambda_m$
\[ |G_{z, \Lambda_m} (0, v)| \leq 4 \epsilon^{\frac{1}{2M}} / |\delta| \leq 4 \epsilon^{\frac{3}{8M}} ~.\]
Applying the second resolvent identity and then the Combes--Thomas estimate (\ref{eq:ct}), we deduce that for $w \notin \Lambda$
\begin{equation}\label{eq:estG}\begin{split} |G_z(0, w)| &\leq \sum_{v \in \Lambda_m, u \notin \Lambda_m} |G_{z,\Lambda_m}(0, v)| \, |A(v-u)| \, 
|G_z(u, w)| \\
&\leq C_1\left[ |\mathcal V| \epsilon^{\frac{3}{8M}}  + |\Lambda \setminus \mathcal V| \max_v |G_{z,\Lambda_m}(0, v)|  \right] \sum_{u \in \mathbb Z^\nu} |G_z(u, w)| \\
&\leq C_2 |\Lambda| \epsilon^{\frac{3}{8M}} \sum_{u \in \mathbb Z^\nu} \frac{1}{\delta} e^{-c\delta \|u - w\|} 
\leq C_3 |\Lambda| \epsilon^{\frac{3}{8M}}  \delta^{-\nu-1} 
\leq  C_3 |\Lambda| \epsilon^{\frac{1}{8M}}~. \end{split}\end{equation}
Now we plug this bound into the contour integral representation 
\[ e^{itH}(0, w) = - \frac{1}{2\pi i} \oint_{\mathcal C} e^{it z} G_z(0, w) dz~.\]
Choose $\mathcal C$ to be the boundary of the rectangle $|z| \leq \| H\|+1$, $|\Im z| \leq \delta$, oriented counterclockwise. On all the contour $|e^{tz}| \leq e^{t\delta}$. On the side edges of $\mathcal C$,  $|G_z(0, w)|$ is exponentially small by the Combes--Thomas estimate (\ref{eq:ct}), and on the horizontal edges we use (\ref{eq:estG}). The conclusion is that for $0 < t \leq \frac{1}{40 M \delta} |\log \epsilon|$
\[\begin{split} | e^{itH}(0, w) | &\leq C_4 \left[   \delta e^{-c \|w\|} +    |\Lambda| e^{t\delta}  \epsilon^{\frac{1}{8M}} \right] \leq C_5 \left[   e^{-c  \|w\|} \delta +  |\Lambda| \epsilon^{\frac{1}{10 M}} \right]~, \end{split}\]
whence
\[ P_t(w) \leq 2C_5  \left[   e^{-2c  \|w\|} \delta^2 +  |\Lambda|^2 \epsilon^{\frac{1}{5 M}} \right] ~.\]  
This concludes the proof of Theorem~\ref{thm}. \qed

\medskip\noindent
Finally, we prove the following estimate used in Section~\ref{s:appl2}: for $\Lambda =[-C' t, C'  t]^\nu$, with $C'$ sufficiently large,
\begin{equation}\label{eq:subbal3'} 
| e^{it H_{\Lambda}}(0, v)- e^{it H} (0, v) | \leq  C e^{-c \max(\|v\|, t)}~. 
\end{equation}
\begin{proof}[Proof of {(\ref{eq:subbal3'})}]
If $\|v\| \geq Ct$, the statement follows from the ballistic bound (\ref{eq:subbal}). Otherwise, by (\ref{eq:intrepr}) and the second resolvent identity
\[ \begin{split}
| e^{it H_{\Lambda}}(0, v)- e^{it H} (0, v) | &\leq \frac{1}{2\pi} \int_{\mathcal C} |dz| e^{t |\Im z|} |G_{z,\Lambda}(0, v) - G_z(0,v) |  \\
&\leq \int_{\mathcal C} |dz| e^{t |\Im z|} \sum_{u \in \Lambda, w \notin \Lambda}|G_{z,\Lambda}(0, u)| \, |A(u - w) |\, | G_z(w,v) |~.
\end{split} \]
As in the proof of (\ref{eq:subbal}), we choose $\mathcal C$ to be the boundary of the rectangle $|\Im z| \leq 1$, $|\Re z| \leq \|H\|+1$. By the Combes--Thomas estimate (\ref{eq:ct}), 
\[ \sum_{u \in \mathbb Z^\nu} |G_{z,\Lambda}(0, u)| \leq C_1~, \quad \sup_{w \notin \Lambda}| G_z(w,v) | \leq C_2 e^{-c_2 C' t}~, \]
for any $z$ on $\mathcal C$, whereas (\ref{eq:assum}) implies that for any $u \in \mathbb Z^\nu$
\[ \sum_{w \in \mathbb Z^\nu} |A(u-w)| \leq C_3~.\]
Thus
\[ | e^{it H_{\Lambda}}(0, v)- e^{it H} (0, v) | \leq C_4 e^{t - c_2 C't} \leq C_4 e^{-c_4 t} \]
provided that $C'$ is sufficiently large.
\end{proof}

\end{document}